\documentclass[12pt]{article}
\usepackage{xcolor}
\usepackage{url}
\usepackage{wrapfig}
\usepackage{float}
\usepackage[margin = 1in]{geometry}
\usepackage{amsmath}
\usepackage{amssymb}
\usepackage{amsfonts}
\usepackage{amsthm}
\usepackage{mathtools}
\usepackage{enumerate}
\usepackage{verbatim}
\usepackage{physics}
\usepackage{euscript}
\usepackage{tabularx}
\usepackage[ruled,vlined,linesnumbered]{algorithm2e}
\usepackage[utf8]{inputenc}
\usepackage{authblk}
\numberwithin{equation}{section}
\usepackage{enumitem}

\newcommand{\FF}{\mathbb{F}}

\DeclareMathOperator{\Norm}{Norm}
\DeclareMathOperator{\evmap}{ev}

\newtheorem{thm}{Theorem}[section]
\newtheorem{prop}[thm]{Proposition}
\newtheorem{lem}[thm]{Lemma}
\newtheorem{cor}[thm]{Corollary}
\theoremstyle{definition}
\newtheorem{defn}[thm]{Definition}
\newtheorem{rmk}[thm]{Remark}

\newtheorem{ex}[thm]{Example}

\newtheorem*{rmk*}{Remark}
\newtheorem*{ex*}{Example}

\newcommand{\F}{\mathbb{F}}
\newcommand{\N}{\mathbb{N}}

\newcommand{\rmv}[1]{}

\newcommand\restr[2]{{
  \left.\kern-\nulldelimiterspace 
  #1 
  \vphantom{\big|} 
  \right|_{#2} 
  }}

\begin{document}

\title{Curve-lifted codes for local recovery using lines}
\author[1]{Gretchen L. Matthews\thanks{gmatthews@vt.edu }}
\author[2]{Travis Morrison\thanks{tmo@vt.edu}}
\author[3]{Aidan W. Murphy\thanks{Aidan.Murphy@jhuapl.edu \\The work of the first and second authors is supported in part by the Commonwealth Cyber Initiative. The work of the first author is supported by NSF DMS-2201075. This work was performed while the third author was at Virginia Tech. Some results on binary norm-trace codes were presented at ISIT 2021 and in the third author's dissertation.}}
\affil[1,3]{Department of Mathematics\\Virginia Tech}
\affil[2]{Johns Hopkins Applied Physics Laboratory}
\date{}

\maketitle

\begin{abstract}
In this paper, we introduce curve-lifted codes over fields of arbitrary characteristic, inspired by Hermitian-lifted codes over $\F_{2^r}$. These codes are designed for locality and availability, and their particular parameters depend on the choice of curve and its properties. Due to the construction, the numbers of rational points of intersection between curves and lines play a key role. 
To demonstrate that and generate new families of locally recoverable codes (LRCs) with high availabilty, we focus on norm-trace-lifted codes. In some cases, they are easier to define than their Hermitian counterparts and consequently have a better asymptotic bound on the code rate. 
\end{abstract}

\section{Introduction} \label{intro_section}

Algebraic geometry codes were introduced in the 1980s \cite{Goppa_83} and quickly received attention due to the existence of sequences with parameters exceeding the Gilbert-Varshamov bound \cite{TVZ}. These codes are defined by evaluating functions on a smooth projective curve over a finite field at rational points, a construction that is quite flexible allowing for customizations or variants that achieve particular goals. In \cite{barg_tamo_vladut_17}, it was demonstrated how coverings of curves yield locally recoverable codes (also known as LRCs or codes with locality). A code $C$ of length $n$ is said to be {\bf locally recoverable} if there exist recovery sets $R_1, \dots, R_n$ such that for every codeword coordinate $i$,  the $i^{th}$ coordinate $c_i$ of any codeword $c \in C$ may be recovered from the codeword symbols $c\mid_{R_i}$. The cardinality of the largest $R_i$ is called the locality of $C$. A code has availability $t$ if each coordinate has $t$ disjoint recovery sets and is called a code with availability if $t>1$.

Locally recoverable codes, whose study originated in \cite{ gopalan_12, dimakis_12}, are studied due to their applicability in setting such as distributed storage where vast amounts of data are stored across many servers which may be temporarily offline (and  is modeled as an erasure). To limit network traffic involved in recovery, it is desirable that each coordinate (or server) can be recovered using information from a small subset all other coordinates (or the rest of the network). Moreover, it is useful to have multiple ways in which information can be recovered, so that if a coordinate is lost, the ability to recover does not depend on the availability of some other coordinate which might also be lost. 

Algebraic geometry codes can be adapted to define locally recoverable codes, as noted in \cite{barg_tamo_vladut_17, fiber}. 
In this paper, we define curve-lifted codes, a generalization of the Hermitian-lifted codes \cite{lifted}, for a projective curve $\mathcal X$ over a finite field $\F_q$. Curve-lifted codes are evaluation codes in which codewords are determined by evaluating particular functions at affine points on the curve. The functions to evaluate depend on what we will refer to as an intersection number: given a curve $\mathcal X$ and a collection of lines $\mathbb L$, the intersection number is   
$$ \min \left\{  \mid \left( L \cap \mathcal X \right) \left( \F_q \right) \mid : L \in \mathbb L \right\}.$$ Codewords arise from the evaluation of rational functions $f$ on $X$ that restrict to low-degree polynomials on $\left( L \cap \mathcal X \right) \left( \F_q \right)$. Recovery sets for a position associated with a point $P$ consist of collections of other points of intersection between a line through $P$ and $\mathcal X$. The parameters of such codes depend on the collection $\mathbb L$ of lines selected and the number of $\F_q$-rational points that lie on both the curve $\mathcal X$ and a line $L \in \mathbb L$. When $\mathbb L$ is the set of all lines in affine space over $\F_q$, we refer to such a number as an intersection number. 

Intersection numbers for particular curves can be challenging to determine. We demonstrate progress in this direction for the 
{\bf  norm-trace curve} which is defined by
\[
\mathcal X_{q,r}:   y^{q^{r-1}}+ y^{q^{r-2}}+ \dots +y^q + y = x^{\frac{q^r-1}{q-1}}
\]
over the field $\F_{q^r}$ with $q^r$ elements. Taking $r=2$ gives the {\bf Hermitian curve} $H_q: y^q+y=x^{q+1}$ over $\F_{q^2}$. This allows us to give important instances of the curve-lifted construction, resulting in norm-trace-lifted codes.

Curve-lifted codes and norm-trace-lifted codes are strongly inspired by Hermitian-lifted codes. 
We note some key distinctions between the binary Hermitian case ($r=2$) and the more general norm-trace one in which $r>2$. The functions which give rise to codewords can be described explicitly for the case $r>2$, while much of the effort in studying the Hermitian case is devoted to finding enough functions to get a positive rate. It should be mentioned that lifted codes are not traditional algebraic geometry codes; while the codewords are of a similar form, obtained by evaluating functions at rational points, the space of functions is not a Riemann-Roch space. 

Determining the functions to evaluate to define a curve-lifted code is a crucial task. As we will see in the norm-trace case when $r>2$, we must first determine the appropriate locality, which depends on the number of rational points of intersection between the curve $\mathcal{X}_{q,r}$ and a line  $\F_{q^r}$. That is a primary contribution of this work. It pays off in that once this is determined, we can immediate specify enough functions which yield codewords in the norm-trace-lifted code to prove that this family of codes has positive rate even as the code length grows, for a fixed characteristic. This is in contrast to the Hermitian case (meaning $r=2$) where it is an open problem to explicitly describe all functions that give rise to codewords. 

This paper is organized as follows. Section \ref{prelim_section} reviews the necessary terminology and the Hermitian-lifted codes. We define curve-lifted codes in Section \ref{curve_lifted_section} and give examples of them. 
The cardinalities of intersections between the norm-trace curves and lines are found in Section \ref{intersection_section}, laying a foundation to define an array of norm-trace-lifted codes over arbitrary fields in Section \ref{codes_section}. The paper ends with a conclusion in Section \ref{conclusion_section}.

\section{Preliminaries} \label{prelim_section}

In this section, we review notation to be used throughout the paper as well as the background on other curve-lifted codes in the literature. Throughout, we let $q$ be a power of a prime and $r\geq 2$ be an integer. The finite field with $q$ elements is denoted $\F_q$, and the multiplicative group of its nonzero elements is denoted $\F_q^\times$.  The set of nonnegative integers is denoted by $\N$, and for a positive integer $n$, $[n]\coloneqq \{1, \dots, n\}$.

An {\bf $[n,k,d]$ linear code} $C$ over a finite field $\F_q$ is a $k$-dimensional $\F_q$-subspace of $\F_q^n$ in which any two distinct elements (called codewords) differ in at least $d$ coordinates. Such a code has {\bf length} $n$, {\bf dimension} $k$, and {\bf minimum distance} $d$. Any $d-1$ erasures in a received word may be recovered by accessing the remaining $n-d+1$ coordinates. We are interested in recovering erasures by accessing a small number of coordinates. As is standard in the erasure recovery model, we assume a received word $w$ have the form $w \in \F_q^n \cup \{ ? \}$ where there exists a codeword $c \in C$ such that $w_i \in \{ c_i, ? \}$ for all $i \in [n]$. All codes considered in this paper are linear, so we use the terms code and linear code interchangeably. 

A code $C$ of length $n$ over $\F_q$ has {\bf locality} $s$ if and only if for all $i \in  \{ 1,\dots,n\}$, there exists 
$$R_i \subseteq [n] \backslash \{ i \}$$ 
where $|R_i| \leq s$ and for all codewords $c \in C$,
$$c_i = \varphi_i (c\mid_{R_i})$$
for some function $\varphi_i : \F_q^s \rightarrow \F_q$. Note that a code with locality $s$ clearly has locality $s'$ for all $s' \geq s$. However, we typically use the term to refer to $\max \{ |R_i|: i \in [n] \}$.  The set $R_i$ is called a {\bf recovery set} for $i$. The set $\overline{R_i} \coloneqq  R_i \cup \{ i \}$ is called a {\bf repair group} for $i$. 
We say that $C$ has {\bf availability} $t$ if there exist recovery sets 
$$R_{i,1},\dots,R_{i,t} \subseteq [n] \backslash \{ i \}$$ 
for each index $i \in [n]$, such that 
$$R_{i,j} \cap R_{i,j'} = \emptyset$$ 
for all $j \neq j'$. 

The codes considered in this paper are reminiscent of algebraic geometry codes, in that they are defined using rational points and functions on curves over finite fields. 
An {\bf algebraic geometry code} $C(D,G)$ of length $n$ over a finite field $\F$ is defined by fixing a curve $\mathcal X$ over $\F$ along with divisors $G$ and $D$ on $\mathcal X$ so that their supports are disjoint and $D$ is the sum of $\F$-rational points $P_1, \dots, P_n$. Each codeword is of the form $(f(P_1), \dots, f(P_n))$ where $f$ is a function in the Riemann-Roch space of $G$. If $D$ is not specified, it is taken to be the sum of all $\F$-rational points other than those in the support of $G$. 
For example, Reed-Solomon codes  are algebraic geometry codes on projective lines. The next most commonly studied algebraic geometry codes are Hermitian codes. Hermitian codes are special cases of norm-trace codes, meaning algebraic geometry codes defined on norm-trace curves. Norm-trace codes were first defined by Geil \cite{geil_03}. They are defined by evaluating functions from Riemann-Roch spaces at rational points on $$\mathcal X_{q,r}: N(x)=Tr(y)$$ over the field $\F_{q^r}$ with $q^r$ elements; here and throughout, the norm and trace will be considered with respect to the extension $\F_{q^r}/\F_q$, so that for any $a \in  \F_{q^r}$, they are 
$$\Norm(a) \coloneqq  a^{\frac{q^r-1}{q-1}}, \Tr(a)\coloneqq \sum_{i=0}^{r-1} 
a^{q^{i}} \in \F_q.$$
 Note that $\mathcal X_{q,r}$ has genus $$g=\frac{1}{2}\left( \frac{q^r-1}{q-1} -1\right)\left( q^{r-1}-1\right).$$ 
 Let $\mathcal X_{q,r}(\FF_{q^r})$ denote the set of $\F_{q^r}$-rational point on the curve $\mathcal X_{q,r}$. For each $a \in \F_{q^r}$, there are $q^{r-1}$ elements $b \in \F_{q^r}$ such that $\Norm(a)=\Tr(b)$, each giving rise to a rational point $(a,b) \in  \mathcal X_{q,r}(\FF_{q^r})$, referred to as an affine point of $ \mathcal X_{q,r}$. In addition, there is a unique point at infinity, $P_{\infty} \in \mathcal X_{q,r}(\FF_{q^r})$. Hence, $\left| \mathcal X_{q,r}(\FF_{q^r}) \right| = q^rq^{r-1}+1=q^{2r-1}+1$. Consider the vector space of functions $V \subseteq \cup_{m=0}^{\infty} \mathcal L(mP_{\infty})$ with no poles other than at $P_{\infty}$, where $P_{\infty}\coloneqq (0:1:0)$ is the unique point at infinity on $\mathcal{X}_{q,r}$.  It is worth noting that the Riemann-Roch space of the divisor $mP_{\infty}$ is 
\begin{equation} \label{nt_RR}
\mathcal L(mP_{\infty}) = \left< x^i y^j : 0 \leq i \leq \frac{q^r-1}{q-1}-1, iq^{r-1}+j\frac{q^r-1}{q-1} \leq m\right> \subseteq \F_{q^r}[x,y].
\end{equation}
We denote the set of polynomials with coefficients in $\F_{q^r}$ and indeterminates $x, y$ of total degree at most $k$ by $\F_{q^r}[x,y]_{\leq k}$. A Hermitian code is an algebraic geometry code over the Hermitian curve. For convenience, we will identify rational functions of the form  $\frac{f(x,y)}{z^{\deg f}}$ where $f(x,y) \in \F[x,y]$ with the polynomial $f(x,y) \in \F[x,y]$.

A Hermitian code with $m \leq q^2-1$ over $\F_{q^2}$ has locality $q$ and availability $q^2-1$ \cite{lifted}. To see this, note that each affine $\F_{q^2}$-rational point $P$ on the Hermitian curve $\mathcal H_q$ has the property that any non-tangent line to $\mathcal H_q$ intersects the curve in $q+1$ $\F_{q^2}$-rational points and there are $q^2-1$ such lines. Because any function $f$ in the Riemann-Roch space $\mathcal L(mP_{\infty})$ may be expressed as $f(x,y)=\sum_{i=0}^{q-1} \sum_{j=0}^{\lfloor \frac{m-iq}{q+1} \rfloor} x^i y^j$, such a function restricted to a non-tangent line can be viewed as a univariate polynomial of degree at most $q-1$. To recover an erasure at point $P_{ab}$, one may treat the word corresponding to the $\F_{q^2}$-rational points on any line through $P_{ab}$ as a Reed-Solomon codeword. Thus, the set of $\F_{q^2}$-rational points on any non-tangent line through $P_{ab}$, other than $P_{ab}$ itself, form a recovery set for the coordinate corresponding to $P_{ab}$, demonstrating that the Hermitian code $C(mP_{\infty}, D)$  over $\F_{q^2}$ has locality $q$ and availability $q^2-1$. Unfortunately, the rate of these codes approaches $0$ as $q$ goes to infinity. 

Lifting is a mechanism introduced to increase the rate of codes while maintaining desirable properties. Hermitian-lifted codes over binary fields were introduced in \cite{lifted} and yield a family of codes with a positive lower bound on the rate as $q$ goes to infinity. More recent work with an improved bound on the rate of Hermitian-lifted codes appears in \cite{allen2023improving}.

\section{Curve-lifted codes} \label{curve_lifted_section}

In this section, we present curve-lifted codes, which are evaluation codes whose codewords arise from functions that restrict to low degree polynomials on a collection of lines through the curve. In the first subsection, we develop the notion in this section and consider some examples. As we will see, intersection numbers will be a necessary ingredient, and they will be further explored in later sections. In the second subsection, we drill down to better describe the defining sets of functions for particular curves. 

\subsection{Construction}

Consider a projective curve $\mathcal X$ given by $F(x,y)=0$ over the finite field $\F_{q}$. Take an $\F_{q}$-rational point $P \in \mathcal X(\F_{q})$  and enumerate the other $\F_{q}$-rational points on $\mathcal X$: $P_1,\dots, P_n$. Set $D=P_1+\dots+P_n$, and let $V \subseteq  \F_{q}(\mathcal X)$ be a set of rational function on $\mathcal X$ with no poles among $P_1, \dots, P_n$. Consider the map
$$
\begin{array}{lccc}
\evmap_D \colon &V &\to &\F_{q}^n \\
&f &\mapsto &(f(P_1), \dots, f(P_n)).
\end{array}.
$$
We define an {\bf evaluation code} $C(V, D)$ to be the image of the map
$\evmap_D$. Its properties will depend on $V$ and $\{ P_1, \dots, P_n \}$. 

We will be interested in non-horizontal lines $$L_{\alpha,\beta}(x,y):y=\alpha x + \beta$$ with $\alpha,\beta \in \FF_{q}$, so $\alpha \neq 0$. Let $$\mathbb L_{q}\coloneqq  \left\{ L_{\alpha,\beta}: \alpha, \beta \in \F_{q}, \alpha \neq 0 \right\}.$$We consider which rational functions $f \in \F_q(\mathcal X)$ restrict to certain polynomials  on the intersections of particular lines and the curve $\mathcal X$. To do so, let 
$$m_{\alpha, \beta, F}(x)\coloneqq F(x,\alpha x + \beta).$$
When $F$ is clear from the context, we may write $m_{\alpha, \beta}(x)$ for $m_{\alpha, \beta,F}(x)$. It is relevant to consider functions $f$ on $\mathcal X$ modulo the polynomial $m_{\alpha, \beta}$. More precisely, we use the following notion. 
\begin{defn}
For a polynomial $f(t) \in \F_{q}[t]$, define $\bar{f}_{\alpha, \beta}(t)$ to be the remainder resulting upon division of $f(t)$ by $m_{\alpha, \beta}(t)$; that is, 
$$\bar{f}_{\alpha, \beta}(t) \coloneqq f(t) \text{ mod } m_{\alpha, \beta}(t).$$
Set
$$\deg_{\alpha, \beta}(f) \coloneqq \deg(\bar{f}_{\alpha, \beta}(t)).$$
\end{defn}
Note that
\[
\deg \left( \bar{f}_{\alpha, \beta}(t) \right) \leq \deg \left( m_{\alpha, \beta} \right)-1
\]
for all $f \in \F_{q^r}[t]$. We also write $f \equiv_{\alpha,\beta} h$ to mean $f \equiv h \mod m_{\alpha,\beta}$, omitting subscripts if they are clear from the context.

Consider for each point $P_i$, $i \in [n]$, the set of lines $\mathbb L_i$ containing the point $P_i$. Let $$B \leq \min \left\{  \mid \left( L \cap \mathcal{X}\right)(\F_{q}) \mid : L \in \mathbb L_i, i \in [n] \right\}.$$
Notice that any line $L$ through $P_i$ intersects $\mathcal{X}$ at at least $B-1$ other $\F_q$-points. 

\begin{defn} \label{curve_lifted_def}
Given  a curve $\mathcal X$ over $\F_{q}$ with divisor $D\coloneqq P_1+\dots +P_n$ supported by $n$ distinct $\F_q$-rational points, a collection of lines $\mathbb L \subseteq \mathbb L_q$, and an integer $B$, the associated  curve-lifted code is $C(D, \mathcal F_{\mathbb L, B})$ where 
\begin{equation} 
\mathcal{F}_{\mathbb L, B} \coloneqq  \left\lbrace f \in \F_{q}\left(\mathcal X \right) : \exists g \in \F_{q}[t]_{\leq B-2} \text{ with }  f \circ L \equiv g \text{ }\forall L \in \mathbb{L}  \right\rbrace.
\end{equation}
\end{defn}

According to Definition \ref{curve_lifted_def}, the codewords in a curve-lifted code are obtained by evaluating at each point in the support of $D$ functions which restrict on all lines in $\mathbb L$ to low-degree polynomials. In the next result, we see that this construction provides locality and availability. 

\begin{prop} \label{prop_B}
The curve-lifted code $C(D, \mathcal F_{\mathbb L, B})$ is a  code of length $n \leq |\mathcal X \left( \F_{q} \right) |$ over $\F_{q}$ with locality $B-1$ and availability $q-1$. 
\end{prop}

\begin{proof}
For $i \in [n]$, consider the set of lines $\mathbb L_i$ through the $\F_{q}$-rational point point $P_i$ on $\mathcal{X}$. By definition, any  line $L \in \mathbb L$ that contains $P_i$ intersects $\mathcal{X}$ in at least $B-1$ other points among the $P_j, j \in [n]\setminus \{ i \}$. Let $R_{i,L} \subseteq \left( L \cap \mathcal X\right) \left( \F_{q} \right) \setminus \{ P_i \}$
such that $| \mathcal R_{i,L} |  = B -1$. 

Consider a received word $w$ resulting from $\evmap(f)$ in which there is an erasure in the coordinate corresponding  $P_i$. We claim that $R_{i,L}$ is a recovery set for position $i$, for all $L \in \mathbb L_i$. To demonstrate this fact, we must determine from $R_{i,L}$ the value $f(P_i)$. Observe that for each of the points in the set  $R_{i,L}\coloneqq \left\{ P_{j_1}, \dots P_{j_{B-1}} \right\}$, the value $f(P_{j_t})$ is known. Since $f\mid_L = g$, the values $g \left( P_{j_1}\right), \dots, g \left( P_{j_{B-1}}\right)$  are known. Because $\deg g \leq B-2$, the polynomial $g$ may be found by interpolation using the $B-1$ values $g \left( P_{j_1}\right),  \dots, g \left( P_{j_{B-1}}\right)$. Then $f(P_i)=g(P_i)$. Hence $R_{i,L}$ is a recovery set for $i$. 
Moreover, the intersection of any two such lines $L$ and $L'$ satisfies 
$$L \cap L' = \{ P_i\}.$$ Thus, the sets 
$R_{i,L}$, $L \in \mathbb L_i$ are disjoint recovery sets indicating that $C$ has availabilty $q-1$.
\end{proof}

\rmv{
The dimension of the code $C(D, \mathcal F_{\mathbb L, B})$ depends on $\mathcal F_{\mathbb L, B}$, which is influenced by the set $\mathcal L$ and the value $B$.}

\begin{ex}
Consider taking $\mathcal X = \mathcal H_q$, the Hermitian curve over $\F_{q^2}$ where $q$
is even. According to \cite{lifted}, $B=q+1$. Hence, Proposition \ref{prop_B} states such codes have length $q^3$, locality $B-1=q$, and availability $q^2-1$. These are precisely the Hermitian-lifted codes considered in \cite{lifted}. 
\end{ex}

\begin{ex} \label{ex_q_3}
In this example, we consider norm-trace-lifted codes over fields $\F_{3^r}$ for small values of $r$. When $r=3$, we have the curve $$\mathcal X_{3,3}: y^9+y^3+y=x^{13}$$  which has genus $48$ and $243$ $\F_{27}$-rational points other than $P_{\infty}$. Using \cite{magma}, we see that each line in affine space over $\F_{27}$ intersects the curve in either $7$, $10$, or $13$ $\F_{27}$-rational points. Taking 
$$B = \min \{ 7, 10, 13 \} =7$$ and $\mathbb L = \mathbb L_{27}$ gives 
$$
\mathcal{F}_{\mathbb L_{27}, 7} \coloneqq  \left\lbrace f \in \F_{27}\left(\mathcal X_{3,3} \right) : \exists g \in \F_{27}[t]_{\leq 5} \text{ with }  f \circ L \equiv g \text{ }\forall L \in \mathbb{L}  \right\rbrace
$$
and the code $C(D, \mathcal F_{\mathbb L, 7})$ of length $243$ with locality $6$ and availability $3^3-1=26$. Codewords are of the form $\evmap_D(f)$ where $\deg f\mid_{(L \cap \mathcal X_{3,3})(\F_{27})} \leq 5$ for all lines $L$ over $\F_{27}$.  As a result, an erasure can be recovered by utilizing only $6$ of the other $242$ coordinates and in $26$ different (disjoint) ways, as each non-horizontal line forms a repair group for each point of intersection with $\mathcal X_{3,3}$. We may also note that $$\left< x^ay^b: a + b \leq 5 \right> \subseteq \mathcal{F}_{\mathbb L_{27}, 7},$$ so $$\left< \evmap_D(x^ay^b) : a + b \leq 5 \right> \subseteq C(D,\mathcal{F}_{\mathbb L_{27}, 7}).$$ The containment may be strict as in the Hermitian case; it will be further explored in Section \ref{codes_section}. 

We may also consider taking a proper subset of lines. Calculating intersection numbers \cite{magma}, we see that 
for every point $(a,b) \in  X_{3,3}(\F_{27})$ with $a \neq 0$, there are
    \begin{enumerate}[label=\alph*.,start=1]
\item $6$ lines through $(a,b)$ meeting $\mathcal X_{3,3}$ in $13$ points
\item  $10$ lines through $(a,b)$ meeting $\mathcal X_{3,3}$ in $7$ points
\item 10 lines through $(a,b)$ meeting $\mathcal X_{3,3}$ in $10$ points. 
    \end{enumerate}
For the points $(0,b) \in  X_{3,3}(\F_{27})$, there are 
\begin{enumerate}[label=\alph*.,start=4]
\item $13$ lines through $(0,b)$ meeting $\mathcal X_{3,3}$ in $13$ points
\item $13$ lines through $(0,b)$ meeting $\mathcal X_{3,3}$ in $7$ points. 
\end{enumerate}
Taking $\mathbb L$ to be the sets of lines in a. and d. above and setting $B=13$, Proposition \ref{prop_B} applies to give a code $C(D,\mathcal F_{\mathcal L, 13})$ which has length $243$, locality $12$, and availability $\min \{ 6, 13\}=6$. Notice that the code $C(D,\mathcal F_{\mathcal L, 13})$ includes codewords defined by functions that become polynomials of degree at most $11$ when restricted to the lines in a. and d. In particular, 
$$\left< \evmap_D(x^ay^b) : a + b \leq 11 \right> \subseteq C(D,\mathcal{F}_{\mathbb L_{27}, 13}).$$ This suggests that taking only lines which intersect the curve in more points ($13$ opposed to $7$) yields codes of larger dimension, a fact that will be considered in Section \ref{codes_section}.

Taking instead $r=4$, we have the curve $$\mathcal X_{3,4}: y^{27}+y^9+y^3+y=x^{40}$$ over $\F_{81}$  which has genus $507$ and $2187$ affine $\F_{81}$-rational points. Computations \cite{magma} indicate that each line in affine space of $\F_{81}$ intersects $\mathcal X_{3,4}$ in $22$, $28$, or $31$ $\F_{81}$-rational points. Thus, we take 
$$B= \min \left\{ 22, 28, 31 \right\}=22$$
and note that
$$
\left< x^ay^b: a + b \leq 20 \right> \subseteq
\mathcal{F}_{\mathbb L_{81}, 22} = \left\lbrace f \in \F_{81}\left(\mathcal X_{3,4} \right) : \exists g \in \F_{81}[t]_{\leq 20} \text{ with }  f \circ L \equiv g \text{ }\forall L \in \mathbb{L}  \right\rbrace.
$$
Hence, $C(D,\mathcal F_{\mathbb L_{81},22})$ is a code over $\F_{81}$ of length $2187$, locality $21$, and availability $80$. It follows that an erasure can be recovered by utilizing only $20$ of the $2186$ other coordinates and in $80$ different (disjoint) ways.
\end{ex}

\begin{ex} \label{ex_quot_herm}
In this example, we consider the curve $\mathcal X$ given by $y^8+y=x^3$ over $\F_{64}$, from the first family of non-classical curves described by Schmidt \cite{Schmidt}. One may note that $\mathcal X$ is maximal \cite{Garcia_Vianna} and a so-called Castle curve \cite{Castle}. Note that $\mathcal X$  which has genus $7$ and $176$ $\F_{64}$-rational points other than $P_{\infty}$. Using \cite{magma}, we see that each line in affine space over $\F_{64}$ intersects the curve in either $1$, $2$, $4$, or $7$ $\F_{64}$-rational points. Taking 
$$B = \min \{ 1, 2, 4, 7 \} =1$$ and $\mathbb L = \mathbb L_{64}$ gives 
$$
\mathcal{F}_{\mathbb L_{64}, 1} =\emptyset.
 $$
Thus, to obtain curve-lifted codes on $\mathcal X$, we must take a proper subset of lines in $\mathbb L_{64}$. Calculating intersection numbers \cite{magma}, we see that
there are $168$ points $P_1, \dots, P_{168}$ such that for each $P_i$, $i \in [168]$, there are \begin{enumerate}[label=\alph*.,start=1]
\item $3$ lines whose only point of intersection with the curve $\mathcal X$ in $P_i$, 
\item  $12$ lines  through $P_i$ meeting $\mathcal X$ in 2 points, 
\item  $11$ lines  through $P_i$ meeting $\mathcal X$ in 3 points,
\item  $30$ lines  through $P_i$ meeting $\mathcal X$ in 4 points, and 
\item  $7$ lines  through $P_i$ meeting $\mathcal X$ in 7 points.
\end{enumerate}
In addition, there are $8$ points $P_{169}, \dots, P_{176}$ such that for each $P_i$, $i \in \left\{ 169, \dots, 176 \right\}$, there are \begin{enumerate}[label=\alph*.,start=6]
\item $21$ lines  through $P_i$ meeting $\mathcal X$ in 3 points and 
\item $42$ lines meeting the curve $\mathcal X$ in $4$ points. 
\end{enumerate}
To design a curve-lifted code on $\mathcal X$ with locality $3$, one could take the set $\mathbb L$ to consist of those lines in d. and g. above. In this case, the code $C(P_1+\dots+P_{176}, \mathcal F_{\mathbb L, 4})$  has length $176$ over $\F_{64}$, locality $3$, and availability $\min \{ 30, 42 \} = 30$. It is formed by taking those functions $f \in \mathbb F_{64}(\mathcal X)$ that reduce to quadratics on the intersection of each line $L \in \mathbb L$ and the curve $\mathcal X$. 

To increase the dimension, we might consider increasing the locality. However, there are no lines  which contain any of the points $P_i$, $i \in \{ 169, \dots, 176 \}$, and intersect the curve in more than $4$ points. Consequently, Definition \ref{curve_lifted_def} does not support the design of a length $176$ code over $\F_{64}$ which has locality greater than $3$. Even so, we note that each of the points $P_i$, $i \in [168]$, lies on $7$ lines that intersect the curve in $7$ points as specified in e. above. The Proposition \ref{prop_B} implies that 
$C(P_1+\dots+P_{168}, \mathcal F_{\mathbb L, 7})$  has length $168$ over $\F_{64}$, locality $5$, and availability $7$.
\end{ex}

As these examples show, a key component of these curve-lifted constructions is the determination of the number of points of intersection between the curve and lines. In Examples \ref{ex_q_3} and \ref{ex_quot_herm}, computational tools were used to determine them for particular  curves over specific fields. In order to determine infinite families of such codes, we need more sophisticated theoretical tools. We will demonstrate that in the next section where we bound intersection numbers of sets of lines on the norm-trace curve. 

One may also note that the traditional code parameters dimension and minimum distance are absent in Proposition \ref{prop_B}. Because locally recoverable codes are designed for erasure recovery using small sets of other coordinates, the minimum distance is not as relevant as in standard error correction or recovery of collections of erasures using the entire received word. However, rates for families of LRCs provide a useful gauge of their capabilities. To examine code rates, we will also need the more information on particular curves. We make some headway on this front in the next subsection.

\begin{defn}
A monomial $M_{a,b}(x,y)$ is said to be {\bf good} for $\mathcal F_{\mathbb L, B}$  if for all lines $L_{\alpha, \beta} \in \mathbb{L}$,
$$\deg_{\alpha, \beta}(M_{a,b} \circ L_{\alpha, \beta}) \leq B-2.$$ 
\end{defn}
We may simply say that a monomial is good if the set of lines $\mathbb L$ and integer $B$ are clear from the context. Some monomials are good regardless of the choice of $\mathbb L$. For instance, $M_{a,b}(x,y)$ is good for $\mathcal{F}_{\mathbb L, B}$ for each 
$(a,b) \in \N^2$ with $a+b \leq B-2$, for all $\mathbb L \subseteq \mathbb L_q$

\subsection{Sporadic monomials}

In this subsection, we study monomials $x^ay^b$ which are good but not simply because $a+b$ is small enough (as mentioned above). This notion is made precise in the following definition.

\begin{defn}
A monomial $x^a y^b$ is called {\bf sporadic} for $\mathcal{F}_{\mathbb L, B}$ if it is good for $\mathcal{F}_{\mathbb L, B}$ and $a + b \geq B-1$. A monomial good for $\mathcal{F}_{\mathbb L, B}$ is called {\bf typical} if it is not sporadic. 
\end{defn}

\begin{ex}
Recall that Hermitian-lifted codes have locality $q$ and we may think of them as $C(D, \mathcal F_{\mathbb L_{q^2}, q+1})$ where $D=P_1+\dots+P_{q^3}$ is supported by the $\F_{q^2}$-rational points on $X_{q,2}: y^q+y=x^{q+1}$. The set of typical monomials is
$$\left\{ x^ay^b : a+b \leq q-1 \right\}.$$

Loosely speaking, Hermitian-lifted codes are defined with two sets of  monomials $x^a y^b$: those with $a + b \leq q-1$, meaning degree less than the locality which are always good, and some with $a + b \geq q$ that happen to reduce to those of degree less than locality on all lines. The monomials in the latter set were called sporadic, since their behavior is not yet fully understood, meaning to date, only some of them have been described explicitly. For instance, according to \cite[Theorem 10]{lifted}, 
$$\left\{ x^ay^b: \begin{array}{l} a \leq q-1, b \leq q^2-1, a+b\geq q, \exists l, i \in [l], 0 \leq s \leq i-1  \textnormal{ so that }
b=wq+b' \\ \textnormal{ with } b'< 2^{l-1}, 2^i \mid w,  a < 2^{l-1}, \textnormal{ and no } 2^s \textnormal{ term in binary expansions of } a \textnormal{ and }b' \end{array} \right\}$$ is a subset of sporadic monomials. 
\end{ex}

In the Hermitian-lifted case, accounting for the sporadic monomials is necessary to determine that the codes have a rate bounded away from $0$ as the code length grows, as demonstrated in \cite{lifted}. We will see that the $r>2$ case for the norm-trace-lifted codes is quite different. While there may be sporadic monomials, due to the larger locality (or $B$ value used), enough typical monomials may be found to demonstrate an even better bound on the asymptotic rate.  To better understand the functions that define codewords of curve-lifted codes, we will use the following observation.

\begin{lem}\label{lemma:monomial_remainder_degree}
    Let $f(x)=x^m + g(x)$ where $\deg g < m$. Assume $d\geq m$. Then the remainder of $x^d$ after division by $f$ has degree at least $\deg g$. 
\end{lem}
\begin{proof}
    Consider the set of integers $A=\{k\geq 1:d+(k-1)(\deg g- m)\geq m\}$. 
    This set contains $k=1$, so it is nonempty. Since $\deg g-m<0$, 
    the sequence $d+(k-1)(\deg g-m)$, $k=1,2,\ldots$ is a strictly decreasing sequence of integers, so $A$ contains a maximal element. Let $k$ be a maximal element of $A$. If we
     write $x^d=q(x)f(x)+r(x)$ with $0\leq \deg r<\deg f$, then 
    \[
    q(x)=x^{d-m}-x^{d-2m}g+x^{d-3m}g^2-\cdots+(-1)^{k-1}x^{d-km}g^{k-1}
    \]
    and $r(x)=(-1)^kx^{d-km}g^k$. Since $k$ is the smallest integer 
    such that $\deg r=d+k(\deg g - m)<m$,  we see that $\deg r$ must be the 
    unique integer in the interval $[\deg g, m)$ congruent to $d$ modulo $m-\deg g$. In particular, $\deg r \geq \deg g$.
\end{proof}

The next result provides insight into when we might expect to find sporadic monomials; a variant of which may be found in \cite{Murphy_dissertation},

\begin{prop} \label{no_sporadic}
Consider a curve $\mathcal X$ given by $F(x,y)=0$ over a finite field $\F_q$ and a collection $\mathbb L \subseteq \mathbb L_{q}$ of lines containing $L_{1,0}$. Let $d$ denote the degree of the second highest degree term of $F$. Suppose a curve-lifted code $C(D, \mathcal F_{\mathbb, B})$ is defined on $\mathcal X$ 
for some integer $B$, selected so that each line in $\mathbb L$ intersects $\mathcal{X}$ in at least $B$ $\FF_{q}$-rational rational points. If $B-1 \leq d$, then there are no sporadic good monomials for $\mathcal F$.
\end{prop}

\begin{proof}
Consider a good monomial $M_{a,b}$. Note that the line $L_{1,0}(x) = (x,x)$ gives $(M_{a,b} \circ L_{1,0}) (x) = x^{a+b}$ and $L_{1,0} \in \mathbb L$. The  intersection of the line defined by $x=y$ with $\mathcal{X}$ is cut out by the equation $F(x,x)=0$, i.e. 
\[
m_{1,0}(x)=F(x,x)=0.
\]
We are interested in the degree of $\overline{x^{a+b}}=x^{a+b} \mod m_{1,0}$. Assume $a+b \geq B-1$, meaning  $M_{a,b}$ is sporadic. We consider two cases, depending on the value $a+b$.
First, suppose $a+b < \deg F$. Then 
$$\overline{x^{a+b}}=x^{a+b},$$ since $\deg {x^{a+b}}=a+b < \deg F.$ However, it then follows that 
$$
\deg_{1,0}(x^{a+b})=a+b \geq B-1>B-2
$$
which shows $M_{a,b}$ is not good for $\mathcal F$. Hence, there are no sporadic monomials $x^{a+b}$ good for $\mathcal F$ where $a+b < \deg F$. 


Now suppose $a+b\geq \deg F$. Then 
the remainder of $x^{a+b}$ upon division by $m_{1,0}(x)$ has degree at least $d$  by Lemma~\ref{lemma:monomial_remainder_degree}. 
Consequently, 
$$
\deg_{1,0}(x^{a+b})\geq d \geq B-1 > B-2,
$$
so there are no sporadic monomials $x^{a+b}$ good for $\mathcal F_{\mathbb L, B}$.
\end{proof}

\section{Intersection numbers of norm-trace curves} \label{intersection_section}
In this section, we determine the number of points in $\mathcal X_{q,r}(\FF_{q^r})$ on an intersection of a line $L_{\alpha,\beta}$ where $\alpha\not=0$ and $\alpha,\beta \in \FF_{q^r}$ with the norm-trace curve 
$\mathcal X_{q,r}$. We loosely refer to the number of such points as an intersection number. Intersection numbers will be applied in Section \ref{codes_section} to construct norm-trace-lifted codes over fields of arbitrary characteristic. In particular, they will be used to set the value $B$ as in Proposition \ref{prop_B} and a degree bound which will define an appropriate set of functions to support local recovery with high availability and positive rate. In particular, we will find an integer $B$, depending only on $q$ and $r$, such that for all $L_{\alpha,\beta}\in \mathbb L_{q^r}$,
\[
B \leq \# \left( L_{\alpha, \beta}  \cap \mathcal X_{q,r} \right)(\F_{q^r})
\]
which will ultimately play a role in the locality of the norm-trace-lifted codes.

Given $\alpha, \beta \in \F_{q^r}$, 
it will be useful to consider the polynomial
\[
m_{\alpha, \beta,q,r}(x) \coloneqq   x^{(q^r-1)/(q-1)} - \Tr(\beta)-\sum_{i=0}^{r-1} (\alpha x)^{q^i} \in \F_{q^r}[x]_{\leq \frac{q^r-1}{q-1}},
\]
or $m_{\alpha, \beta}$ for short. Define 
\[
n_{q,r}(\alpha,\beta) \coloneqq \# (L_{\alpha,\beta}\cap \mathcal X_{q,r})(\FF_{q^r}).
\]
First, we observe that if $\alpha$ and $\alpha'$ are nonzero elements with the same norm and $\beta$ and $\beta'$ have the same trace, then $n_{q,r}(\alpha,\beta)=n_{q,r}(\alpha',\beta')$.

\begin{lem} \label{nt_depend_prop}
For any $\alpha \in \FF_{q^r}^{\times}$ and $\beta \in \FF_{q^r}$, $n_{q,r}(\alpha,\beta)$ depends only on $\Tr(\alpha)$ and $\Norm(\beta)$.
\end{lem}

\begin{proof}
First, plug in $y=\alpha x + \beta$ into the equation for $\mathcal X_{q,r}$ to obtain 
\[
x^{(q^r-1)/(q-1)} = \sum_{i=0}^{r-1} (\alpha x+\beta)^{q^i} = \sum_{i=0}^{r-1} (\alpha x)^{q^i}+\beta^{q^i} = \Tr(\beta)+\sum_{i=0}^{r-1} (\alpha x)^{q^i}.
\]
Then notice that $n_{q,r}(\alpha, \beta)$ is the number of zeros in $\FF_{q^r}$ of the polynomial
$m_{\alpha,\beta}(x)$. 
Making the substitution $t=\alpha^{-1}x$ yields the polynomial
\begin{align*}
g(t) &= m_{\alpha,\beta}(\alpha^{-1}x) \\
&= (\alpha t)^{(q^r-1)/(q-1)} - \Tr(\beta) - \sum_{i=0}^{r-1} t^{q^i} \\
&= \Norm(\alpha)t^{(q^r-1)/(q-1)} - \Tr(\beta) - \sum_{i=0}^{r-1} t^{q^i}. 
\end{align*}
The number of zeros in $\FF_{q^r}$ of $g$ and $m_{\alpha,\beta}$ are the same since the map $t\mapsto \alpha^{-1} x$ is a bijection sending zeros of $m_{\alpha,\beta}$ to zeros of $g$, and the number of zeros of $g$ depends only on $\Tr(\beta)$ and $\Norm(\alpha)$. 
\end{proof}

Next, we apply Lemma \ref{nt_depend_prop} to the case where $q=2$; see also \cite[Lemma 1]{MM_ntl_binary}.

\begin{prop} \label{lemma:intwithlines}
Let $r \geq 2$ and $\alpha, \beta \in \F_{2^r}$ with $\alpha \neq 0$. If $\Tr(\beta) \neq 0$, then $$n_{2,r}(\alpha, \beta) = 2^{r-1} -1.$$ If $\Tr(\beta) = 0$, then $$n_{2,r}(\alpha, \beta) = 2^{r-1} +1.$$ Thus, for any line $L_{\alpha, \beta} \in \mathbb{L}_{2,r}$, the cardinality of its intersection with the norm-trace curve $\mathcal{X}_{2,r}$ over $\F_{2^r}$ is 
 $$| L_{\alpha, \beta}  \cap \mathcal{X}_{2,r}\left( \F_{2^r} \right) |  =2^{r-1} \pm 1.$$
\end{prop}

\begin{proof}
To determine $n_{2,r}(\alpha,\beta)$, according to Lemma \ref{nt_depend_prop}, we need only consider the two cases, depending on $\Tr(\beta)=0$ or $\Tr(\beta)=1$. 

Notice that points in the intersection $L_{\alpha, \beta} \cap \mathcal{X}_{2,r}\left( \F_{2^r} \right)$ correspond to roots of the polynomial
$$x^{2^r-1} - (\alpha x + \beta)^{2^{r-1}} + \cdots + (\alpha x + \beta)^2 + (\alpha x + \beta)$$ which are also elements of $\F_{2^r}$. Because  the roots of $x^{2^r-1}-x$ are precisely the $\gamma \in \F_{2^r}$, the problem of determining $n_{2,r}$ reduces to finding the degree of 
$h(x) = \text{gcd}(m_{\alpha, \beta}(x),x^{2^r}-x)$
 by Freshman's Dream. 

In the case $\Tr(\beta) = 0$, the Euclidean Algorithm reveals $$\text{gcd}(m_{\alpha, \beta}(x),x^{2^r}-x) = \alpha^{2^{r-1}} x^{2^{r-1}+1} + \cdots + \alpha x^2 + x$$ which has degree $2^{r-1}+1$. In the case $\Tr(\beta) = 1$, $$\text{gcd}(m_{\alpha, \beta}(t),x^{2^r}-x) = \alpha^{2^{r-1}} x^{2^{r-1}-1} + \cdots + \alpha,$$
which has degree $2^{r-1}-1$. Therefore, $n_{2,r}(\alpha,\beta) = 2^{r-1} \pm1$.
\end{proof}

For $q \neq 2$, we observe more intricate behavior. We will use results of Moisio and Moisio-Wan, who build on work of Katz~\cite{Katz93}, to establish lower and upper bounds on $n_{q,r}(\alpha,\beta)$. 
For an integer $r\geq 2$ and elements $a,b\in \FF_q$, let 
\[
N_{q,r}(a,b)=\#\{\alpha \in \FF_{q^r}: \Tr(\alpha)=a,\Norm(\alpha)=b\}.
\]

In~\cite{Katz93}, Katz proves the following result on counting elements of $\FF_{q^r}$ with prescribed norm and trace:
\begin{lem}~\cite[Theorem 4]{Katz93}
If $r\geq 2$ and $a,b\in \FF_q^\times$, then 
\[
\left|N_{q,r}(a,b) - \frac{q^{r}-1}{q(q-1)}\right| \leq rq^{(r-2)/2}.
\]
\end{lem}

We will use the following improvement to Katz' result when $a\not=0$, due to Moisio and Wan~\cite{MW10}.

\begin{lem}\cite[Theorem 1.2]{MW10}\label{thm:trnonzero}
    Let $a,b\in \FF_q^{\times}$ and $r\geq 2$. Then 
    \[
    \left| N_{q,r}(a,b)-\frac{q^{r-1}-1}{q-1} \right|\leq (r-1)q^{(r-2)/2}.
    \]
\end{lem}

The previous bound is complemented by the following result of Moisio \cite{Moisio}, which provides a bound $N_{q,r}(a,b)$ when $a=0$. 
\begin{lem}~\cite{Moisio}\label{thm:trzero}
Let $r\geq 2$, $b\in \FF_q^{\times}$, and $d=\gcd(r,q-1)$. We have 
\[
\left|N_{q,r}(0,b)-\frac{q^{r-1}-1}{q-1}\right|\leq (d-1)q^{(r-2)/2}.
\]
\end{lem}

We now use~\cite{Moisio,MW10} to produce lower bounds on $n_{q,r}(\alpha,\beta)$.
\begin{thm} \label{bounds_thm}
Let $r\geq 2$,  $\alpha,\beta\in \FF_{q^r}$ with $\alpha\not=0$, and $d=\gcd(r,q-1)$. If $\Tr(\beta)\not=0$, then 
\[
n_{q,r}(\alpha,\beta) \geq  q^{r-1}-(d-1+(r-1)(q-2))q^{(r-2)/2}-1.
\]
If $\Tr(\beta)=0$, then 
\[
n_{q,r}(\alpha,\beta)\geq  q^{r-1}-(r-1)(q-1)q^{(r-2)/2}.
\]
\end{thm}

\begin{proof}
    Let $a=\Norm(\alpha)^{-1}$ and $b=\Tr(\beta)$.
    Let $g(t)=a^{-1}t^{(q^r-1)/(q-1)}-b-\sum_{i=0}^{r-1}t^{q^i}$. We have that $n_{q,r}(\alpha,\beta)$ is the number of roots of $g$ in $\FF_{q^r}$. Note that the roots of $g$ are in bijection with the set 
    \[
   \bigcup_{t\in \FF_q} \{\gamma \in \FF_{q^r}: \Tr(\gamma) =t \text{ and } \Norm(\gamma) = at+ab\}.
    \]
    Thus the number of roots of $g$ (and hence $n_{q,r}(\alpha,\beta)$) is given by 
    \begin{align*}
    n_{q,r}(\alpha,\beta) &= \#\{\gamma\in \FF_{q^r}:g(\gamma)=0\} \\ 
    &= \sum_{t\in \FF_q} \#\{\gamma \in \FF_{q^r}: \Tr(\gamma) = t \text{ and }   \Norm(\gamma) = at+ab\} \\ 
    &= \sum_{t\in \FF_q} N_{r}(t,at+ab).
    \end{align*}
    First, assume $b\not=0$. When $t=-b$, we have $N_{q,r}(t,at+ab)=N_{q,r}(-b,0)=0$. If $t\not=-b$, then  $at+ab\not=0$ so we may apply Lemma~\ref{thm:trnonzero} to bound $N_{q,r}(t,at+ab)$. Since $\alpha\not=0$, we may apply Lemma~\ref{thm:trzero} to bound $N_{q,r}(0,ab)$. We therefore obtain the desired lower bound for $n_{q,r}(\alpha,\beta)$:
    \begin{align*}
    n_{q,r}(\alpha,\beta) &= N_{q,r}(0,ab) + N_{q,r}(-b,0) + 
    \sum_{\substack{t\in \FF_q \\ t\not=0,-b}} N_{r}(t,at+ab) \\ 
    &= N_{q,r}(0,ab) + \sum_{\substack{t\in \FF_q \\ t\not=0,-b}} N_{r}(t,at+ab) \\ 
    &\geq \frac{q^{r-1}-1}{q-1}  -(d-1)q^{(r-2)/2} + \sum_{\substack{t\in \FF_q \\ t\not=0,-b}} \frac{q^{r-1}-1}{q-1} - (r-1)q^{(r-2)/2} \\ 
    &= \frac{q^{r-1}-1}{q-1}  -(d-1)q^{(r-2)/2} + (q-2) \left(\frac{q^{r-1}-1}{q-1} - (r-1)q^{(r-2)/2}\right) \\ 
    &= \frac{q^{r-1}-1}{q-1}(1+q-2) - q^{(r-2)/2}(d-1+(q-2)(r-1)) \\ 
    &= q^{r-1}-(d-1+(q-2)(r-1))q^{(r-2)/2}-1.
    \end{align*}

    Now assume $b=0$. Then applying Lemma~\ref{thm:trnonzero} to each $N_{q,r}(t,at)$ for $t\not=0$ we get
    \begin{align*}
    n_{q,r}(\alpha, \beta) &= N_{q,r}(0,0) +  
    \sum_{\substack{t\in \FF_q \\ t\not=0}} N_{r}(t,at) \\ 
    &\geq 1 + \sum_{\substack{t\in \FF_q \\ t\not=0}} \frac{q^{r-1}-1}{q-1} - (r-1)q^{(r-2)/2} \\ 
    &= 1 + (q-1)\left(\frac{q^{r-1}-1}{q-1} - (r-1)q^{(r-2)/2}\right) \\ 
    &= q^{r-1}-(r-1)(q-1)q^{(r-2)/2}.
    \end{align*}
\end{proof}

\begin{cor} \label{lower_bound_cor}
For any line
$L_{\alpha, \beta} \in \mathbb L_{q^r}$, the cardinality of its intersection with the norm-trace curve $\mathcal X_{q,r}$ satisfies 
\[
\left| L_{\alpha, \beta}  \cap \mathcal X_{q,r}(\F_{q,r}) \right| \geq q^{r-1}-(r-1)(q-1)q^{(r-2)/2} -1.
\]There are $q^r-1$  lines in $\mathbb L_{q^r}$. Let $d=\gcd(r,q-1)$. For those lines $L_{\alpha, \beta} \in \mathbb L_{q^r}$ with $Tr(\beta) \neq 0$, the cardinality of the intersection of $L_{\alpha,\beta}$ with the norm-trace curve $\mathcal X_{q,r}$ satisfies 
$$
\left| L_{\alpha, \beta}  \cap \mathcal X_{q,r}(\F_{q,r}) \right| \geq q^{r-1}-(r-1)(q-1)q^{(r-2)/2}-1 + q^{\frac{r-2}{2}}
$$
if $r \neq d$. 
There are $q^r - q^{r-1}=q^{r-1}(q-1)$ such lines.
\end{cor}

\begin{proof}
Notice that 
$d=\gcd(r,q-1) \leq r$ . For convenience, let $N=q^{r-1}-(q-1)(r-1)q^{(r-2)/2}-1$, which is the right-hand side of the lower bound on the intersection number for lines $L_{\alpha, \beta}$ with $\Tr(\beta) \neq 0$, and $Z= q^{r-1}-(d-1+(q-2)(r-1))q^{(r-2)/2} - 1$, which is the right-hand side of the lower bound on the intersection number for lines $L_{\alpha, \beta}$ with $\Tr(\beta) =0$. Then for all $\alpha, \beta \in \F_{q^r}$, $n_{q,r}(\alpha, \beta) \geq \min \{ N, Z \}$. Calculating the difference between the values given in the two lower bounds, we see that 
$$N-Z=(r-d)q^{\frac{r-2}{2}}-1 = \begin{cases}
-1 & \textnormal{if } d=r \\
q^{\frac{r-2}{2}} - 1 + t, \textnormal{ for some } t\in \N & \textnormal{otherwise}.
\end{cases}
$$
Therefore, $$N = \begin{cases}
Z-1 & \textnormal{if } d=r \\
Z-1+q^{\frac{r-2}{2}} + t, \textnormal{ for some } t\in \N & \textnormal{otherwise}.
\end{cases}
$$
Consequently, every line $L_{\alpha, \beta}\in \mathbb L$ satisfies 
$$
    n_{q,r}(\alpha,\beta) \geq \min \{ N, Z \} \geq Z-1.
$$    
\end{proof}

We will use the bound in the previous corollary to establish families of norm-trace-lifted codes in arbitrary characteristic. 

\begin{rmk} \label{rmk_r_2_bad_bounds}
Notice Theorem \ref{bounds_thm} and Corollary \ref{lower_bound_cor} provide little to no information in the case $r=2$. Indeed, 
$q^{r-1}-(r-1)(q-1)q^{(r-2)/2}-1=q-(q-1)=1$. 
Moreover, in the case that $q$ is even, $d=\gcd(2,q-1)=1$ and 
 $ q^{r-1}-\left( d-1+(r-1)(q-2)\right)q^{(r-2)/2}-1=q-\left( 1-1+(2-1)(q-2)\right)q^{(2-2)/2}=2.$ Furthermore, if $q$ is odd, then $d=2$ and  $q^{r-1}-\left( d-1+(r-1)(q-2)\right)q^{(r-2)/2}-1=q-\left( 2-1+(2-1)((q-2)\right)q^{(2-2)/2}=1$. Hence, these lower bounds are quite poor, since it is known that the actual value for $n_{q,2}(\alpha, \beta)=q+1$. 

\end{rmk}

\section{Norm-trace-lifted codes} \label{codes_section}

In this section, we combine the results from Sections \ref{curve_lifted_section} and \ref{intersection_section} to define and study norm-trace-lifted codes defined using the norm-trace curve $X_{q,r}: Tr(y)=N(x)$ over $\F_{q^r}$ where $q$ is any prime power.  
In light of Remark \ref{rmk_r_2_bad_bounds} and \cite{lifted}, we restrict our attention to $r>2$. Now, having found a lower bound to take for $B$ in Proposition \ref{prop_B}  on intersection numbers as in Corollary \ref{lower_bound_cor}, we may now consider codes defined by sets of rational functions that reduce on all lines to polynomials of degree at most $B-2$. Also, because the curve itself is specified by a particular equation, we can obtain bounds on the code rates, including some that exceed the comparable Hermitian cases. 

\subsection{Constructions with highest availability}

To obtain codes from the norm-trace curve with highest availability, we use points on every line through a point to form a repair group so as to obtain the largest number of disjoint recovery sets. With that in mind, consider taking $B=B_{q,r}$ where
$$
B_{q,r} \coloneqq 
\begin{cases}
q^{r-1}-1 & \textnormal{if } q=2 \\
q^{r-1}-(r-1)(q-1)q^{\frac{r-2}{2}}-1 & \textnormal{otherwise}
\end{cases}
$$
with the goal of forming a repair group $R_{ab,L}$ for an affine point $P_{ab}\coloneqq (a,b)$ from each line $L \in \mathbb L_{q^r}$ through $P_{ab}$. 
We will use the shorthand notation $\mathcal F_{q,r} \coloneqq \mathcal{F}_{\mathbb L_{q^r}, B_{q,r}}$ so that 
$$
\mathcal F_{q,r}=
\left\lbrace f \in \F_{q^r}[x,y] : \exists g \in \F_{q^r}[t]_{\leq B_{q,r}-2} \text{ with }  f \circ L_{\alpha, \beta} \equiv g \text{ for all } L_{\alpha, \beta} \in \mathbb{L}_{q^r} \right\rbrace.$$ 
To do so, we will take a subset $R_{ab,L} \subseteq L \cap \mathcal{X}_{q,r}(\F_{q^r})$ such that $P_{ab} \in R_{ab,L}$ and $|R_{ab,L} | = B_{q,r}$. Such a subset exists by Corollary  \ref{lower_bound_cor}. Recall that $$m_{\alpha, \beta}(t) \coloneqq  t^{(q^r-1)/(q-1)} - \Tr(\beta)-\sum_{i=0}^{r-1} (\alpha t)^{q^i} \in \F_{q^r}[t]_{\leq \frac{q^r-1}{q-1}}.$$ Note that
\[
\deg \left( \bar{f}_{\alpha, \beta}(t) \right) \leq q^{r-1}+q^{r-2}+\dots+q
\]
for all $f \in \F_{q^r}[t]$, as $\deg \left( m_{\alpha, \beta}(t) \right)= \frac{q^r-1}{q-1}$. 

According to (\ref{nt_RR}), rational functions on $\mathcal X_{q,r}$ with no poles at any of the affine points are elements of $\F_{q^r}[x,y]$, since $\cup_{m \in \N} \mathcal L(mP_{\infty}) \subseteq \F_{q^r}[x,y]$. 

\begin{defn} \label{ntl_code_def}
 The {\bf norm-trace-lifted   code}  defined over $\F_{q^r}$ is $C(D, \mathcal F_{q,r})$, the image of $\mathcal F_{q,r}$ under the evaluation map $\evmap$; that is, 
\[C(D, \mathcal F_{q,r}) \coloneqq \lbrace \evmap_D(f): f \in \mathcal F_{q,r} \rbrace \subseteq \F_{q^r}^{n}.
\]
\end{defn}

Clearly, $C(D, \mathcal F_{q,r})$ is a code of length $n$. To ascertain its dimension, we set out to determine the functions in $\mathcal F_{q,r}$. 
Based on the definition of $\mathcal F_{q,r}$, we are interested in polynomials $f(x,y)$ such that
$$\deg_{\alpha, \beta}(f \circ L_{\alpha, \beta}) \leq  B_{q,r}-2$$ for all $L_{\alpha, \beta} \in \mathbb L_{q^r}$. Recall that $M_{a,b}(x,y)\coloneqq x^ay^b$ where $a, b \in \N$.

 The number of $(a,b) \in \N^2$ with $a+b \leq B$ for some specified positive integer $B$ is $\sum_{a=0}^B B-a+1 =\frac{1}{2}(B+1)(B+2)$. To ensure that the set of such monomials gives rise to an independent set of codewords, we appeal to a result recorded in \cite{Murphy_dissertation}. 

\begin{lem} \cite[Lemma 2.14]{Murphy_dissertation}
\label{monoLemma}
The set of vectors
$$\left\lbrace \evmap(M_{a,b}(x,y)) : 0 \leq a \leq \frac{q^r-1}{q-1}-1, 0 \leq b \leq q^{r-1}-1 \right\rbrace$$
are linearly independent.
\end{lem}

\begin{proof}
Drawing inspiration from the proof of Proposition 5 of \cite{lifted}, we observe that the kernel of the evaluation map $\evmap$ is generated by
$x^{\frac{q^r-1}{q-1}} - y^{q^{r-1}} - \cdots - y^q - y$,
$x^{q^r}-x$,
and
$y^{q^r}-y$. 
Under monomial orderings with $x^{\frac{q^r-1}{q-1}} < y^{q^{r-1}}$, $$\left\{ x^{\frac{q^r-1}{q-1}} - y^{q^{r-1}} - \cdots - y^q - y, x^{q^r}-x \right\}$$ is a Gr\"obner basis for the kernel of the evaluation map, and so the evaluations of $M_{a,b}$ cannot contain any element from the kernel of the evaluation map. Thus, the evaluations of $M_{a,b}$ are linearly independent.
\end{proof}

\begin{thm} \label{code_thm}
The norm-trace-lifted code $C(D, \mathcal F_{q,r})$ over $\F_{q^r}$ is a code of length $q^{2r-1}$, dimension at least $$  \frac{1}{2} q^{r/2-1} \left(q^{r-1}-(q-1) (r-1) q^{r/2-1}+1\right) \left(q^{r/2}-q r+q+r-1\right),$$ locality $q^{r-1}-(r-1)(q-1)q^{\frac{r-2}{2}}-2$, and availability $q^r-1$ with rate approaching $\frac{1}{2q}$ as the code length grows for a fixed characteristic.
\end{thm}

 \begin{proof}
Given an erasure in the coordinate corresponding to an affine point $P_{ab}$ on $\mathcal{X}_{q,r}$, consider the set of lines $\mathcal L$ through $P_{ab}$. According to Corollary \ref{lower_bound_cor}, any  line $L \in \mathcal L$ intersects $\mathcal{X}_{q,r}$ in at least $B_{q,r}-1$ other affine points. Moreover, any function $f \in \mathcal{F}_{q,r}$ has the property that $f_{\mid L} \equiv g$ where $\deg g \leq B_{q,r}-2$. We claim that the $B_{q,r}-1$ affine points in the intersection $L \cap {X}_{q,r} $ other than $P_{ab}$ may be used to interpolate and find $g$; that is, we claim that $L \cap {X}_{q,r} \setminus \{ P_{ab} \}$ is a recovery set for the coordinate associated with $P_{ab}$. Evaluating $g(P_{ab})$ allows for recovery of the erased coordinate using $B_{q,r}-1$ points. Hence, the locality of $C(D, \mathcal F_{q,r})$ with this choice of recovery sets is $B_{q,r}-1= q^{r-1}-(r-1)(q-1)q^{\frac{r-2}{2}}-2$. Because there are $q^r-1$ such lines $L$, $C(D, \mathcal F_{q,r})$ has availability $q^r-1$. 

 Notice that the set $S\coloneqq \left\{ x^ay^b: a+b \leq B_{q,r}-2 \right\}$ is a set of monomials good for $\mathcal{F}_{q,r}$
 of cardinality 
\[
\frac{1}{2}
\left(q^{r-1}-(r-1)(q-1)q^{\frac{r-2}{2}} \right)
\left(q^{r-1}-(r-1)(q-1)q^{\frac{r-2}{2}}+1\right).
\]

According to Lemma \ref{monoLemma}, $S$ is linearly independent, demonstrating that $C(D, \mathcal F_{q,r})$ has dimension at least 
\[
\frac{1}{2}
\left(q^{r-1}-(r-1)(q-1)q^{\frac{r-2}{2}} \right)
\left(q^{r-1}-(r-1)(q-1)q^{\frac{r-2}{2}}+1\right).
\]
and rate
\[
\frac{
\frac{1}{2}
\left(q^{r-1}-(r-1)(q-1)q^{\frac{r-2}{2}} \right)
\left(q^{r-1}-(r-1)(q-1)q^{\frac{r-2}{2}}+1\right)
}{q^{2r-1}} \rightarrow \frac{1}{2q}>0
\]
as $r \rightarrow \infty$. 
 \end{proof}
 
Note that the rate is bounded away from $0$ for fixed characteristic as the degree of the extension grows. Hence, the codes over fields of small characteristic provide the best asymptotic rates. We also observe that taking functions that reduce to low degree polynomials on all lines (rather than just some) provides the largest availability given by the lifted construction, since the lines are the recovery sets. We may also consider fewer lines with more refined intersection numbers, as in the next subsection.

\subsection{Constructions with high availability and larger dimension}

As suggested by Proposition \ref{prop_B}, 
we may adapt the collections of lines $\mathcal L$ used in Definition \ref{ntl_code_def} and the value $B$ to a more refined bound on the intersection numbers. Consider
$$
B_{q,r}'\coloneqq   q^{r-1}-\left( \gcd(r,q-1)-1+(r-1)(q-2)\right)q^{(r-2)/2}-1
$$
and the collection of lines $$\mathbb L_{q^r}' = \left\{ L_{\alpha,\beta}: \Tr(\beta) \neq 0 \right\}.$$ Then set $\mathcal F_{q,r}'\coloneqq \mathcal{F}_{\mathbb L_{q^r}', B_{q,r}'}$ so that 
$$\mathcal{F}_{q,r}' =  \left\lbrace f \in \F_{q^r}[x,y] : \exists g \in \F_{q^r}[t]_{\leq B_{q,r}'-2} \text{ with } f \circ L_{\alpha, \beta} \equiv g \text{ for all } L_{\alpha, \beta} \in \mathbb{L}_{q^r}',  \right\rbrace.$$ We say 
the {\bf refined norm-trace-lifted code} is
$${C}\left(D, \mathcal F_{\mathbb L_{q^r}', B_{q,r}'} \right) = \lbrace \evmap(f) : f \in \mathcal{F}'_{q,r} \rbrace \subseteq \F_{q^r}^{n}.$$ We say that 
$M_{a,b}(x,y)$ is  {good} for $\mathcal{F}'_{q,r}$ if for all lines $L_{\alpha, \beta} \in \mathbb{L}_{q,r}$,
$ \deg_{\alpha, \beta}\left(M_{a,b} \circ L_{\alpha, \beta} \right) \leq B_{q,r}'.$

Notice that $M_{a,b}(x,y)$ is good for $\mathcal{F}'_{q,r}$ for each 
$(a,b) \in \N^2$ with  $a+b \leq B_{q,r}'$. Consequently, we have the following result.

\begin{prop}
The refined norm-trace-lifted code   $C'_{q,r}=\evmap \left( \mathcal{F}'_{q,r}\right)$ is a code over $\F_{q^r}$  of length $q^{2r-1}$,  locality $q^{r-1}-\left( \gcd(r,q-1) -1+(r-1)(q-2)\right)q^{\frac{r-2}{2}}-2$, and availability $q^r-q^{r-1}-2$ with rate bounded away from $0$ as the code length grows for a fixed characteristic.
\end{prop}

\begin{proof}
The proof is similar to that of Theorem \ref{code_thm}.
\end{proof}

In the next subsection, we draw some comparisons between these codes and other families. 

\subsection{Comparisons}

Recall that the norm-trace-lifted codes defined over fields of arbitrary characteristic arise from the set of functions $\mathcal{F}_{q,r}$ which depend on the degree bounds given in Theorem \ref{bounds_thm}. These bounds hold for any $q$. Tighter bounds may  give rise to codes with better parameters. 

We begin this subsection with a demonstration of this. 

\begin{ex}
    Consider a curve-lifted code on the norm-trace curve $\mathcal X_{2,r}$ over $\F_{2^r}$. First consider 
$C(D,\left( \mathcal{F}_{\mathbb L_{2^r}, 2^{r-1}-(r-1)2^{(r-2)/2}} \right))$. 
    Recall that Theorem \ref{bounds_thm} gives 
$$n_{2,r} \geq 2^{r-1}-(r-1)2^{(r-2)/2}$$
whereas
$$n_{2,r}(\alpha, \beta) =2^{r-1} \pm 1,$$ as shown in Proposition \ref{lemma:intwithlines}. 
    Hence, a code with larger dimension is obtained by considering 
    the binary norm-trace-lifted code
$C(D,\left( \mathcal{F}_{\mathbb L_{2^r}, B_{2,r}} \right))$, which was also studied in \cite{MM_ntl_binary}. We see that  $[2^{2r-1},(0.25 - \varepsilon_r) \cdot 2^{2r-1},\geq 2^r]$ code with locality ${2^{r-1}-2}$, availability ${2^r-1}$, and asymptotic rate  $0.25$ \cite[Theorem 3]{MM_ntl_binary}. 

To increase the dimension further, we may take $\mathbb L_{2^r}''\coloneqq \left\{ L_{\alpha, \beta}: Tr(\beta) = 0 \right\}$. In doing so, according to the proof of Proposition \ref{lemma:intwithlines}, we obtain a code $C (D, \mathcal F_{\mathbb L_{2^r}'', 2^{r-1}+1})$. We will see that  
$$
\left< \evmap_B(x^a y^b) : a+b = 2^{r-1}-1, 2^{r-1}-2  \right> \in C (D, \mathcal F_{\mathbb L_{2^r}'', 2^{r-1}+1}) \setminus C(D,\left( \mathcal{F}_{\mathbb L_{2^r}, B_{2,r}} \right)).
$$

Further comparisons are captured in Table \ref{param_table}.
\end{ex}

   \begin{table}[h]
    \begin{tabular}{|l|c|c|c|c|}
        \hline
        & 1-pt norm-trace & HLC & NTLC & RNTLC \\
        \hline
        locality &  $2^{r-1}-2$ & $2^{r / 2}$ & $2^{r-1}-2$ & $2^{r-1}$ \\
        availability & $2^r-1$ & $2^r-1$ & $2^r - 1$ & $2^{r-1} - 1$ \\
        length  & $2^{2r-1}$ & $2^{3r / 2}$ & $2^{2r-1}$ & $2^{2r-1}$ \\
        dimension  & $\leq 2^{2r-4}-1$ & $\geq 0.007 \cdot 2^{3r/2}$ & $(0.25 - \varepsilon_r) \cdot 2^{2r-1}$  & $(0.25 - \varepsilon_r) \cdot 2^{2r-1}$\\
        asymptotic rate  & $\leq \frac{1}{8} - \frac{1}{2^{2r-1}}$ & $\geq 0.007$ & $0.25 $  & $0.25$  \\
        \hline
    \end{tabular}
    \caption{Parameters of one-point norm-trace, Hermitian-lifted, binary norm-trace-lifted, and refined binary norm-trace-lifted codes over $\F_{2^r}$; see also {\cite{Murphy_dissertation}}}
    \label{param_table}
        \end{table}

\begin{rmk}
We recognize that while the binary norm-trace-lifted codes $C(D,\left( \mathcal{F}_{\mathbb L_{2^r}, B_{2,r}} \right))$ have dimension exceeding their counterparts $C(D,\left( \mathcal{F}_{\mathbb L_{2^r}, 2^{r-1}-(r-1)2^{(r-2)/2}} \right))$ defined from the same curve, their asymptotic rates behave similarly. Moreover, there is a tradeoff in locality which is larger for the codes $\evmap\left( \mathcal{F}_{b,2,r} \right)$ than $\evmap\left( \mathcal{F}_{2,r} \right)$, meaning more symbols are needed in order to perform recovery. 
\end{rmk}

It is also worth comparing the codes introduced in this paper with their Hermitian counterparts. A crucial distinction arises if we wish to consider the codes concretely in terms of bases or generator matrices. According to the next result, these are readily available for norm-trace-lifted codes with $r>2$.

\begin{cor} \label{no_spor_nt}
    For $r>2$, the norm-trace-lifted code $C(D, \mathcal F_{q,r})$ over $\F_{q^r}$ is defined exclusively by typical monomials, meaning 
    $$
    \left\{ \evmap \left( x^ay^b \right): a + b \leq B_{q,r} -2 \right\}
    $$ is a basis for $C(D, \mathcal F_{q,r})$.
\end{cor}

\begin{proof}
    The result follows immediately from Proposition \ref{no_sporadic}, since $q^{r-1} \geq B_{q,r}$. 
\end{proof}

Notice that Proposition \ref{no_sporadic} does not apply in the Hermitian case, because $r=2$. In particular, $q \leq B = q+1$. 

\begin{rmk}
        Consider the code $C(D, \mathcal F_{\mathbb L_{q^r}, n_{q,r}})$ on the norm-trace curve $\mathcal X_{q,r}$ over $\F_{q,r}$, where 
$$
n_{q,r}:=\min \left\{ n_{q,r}(\alpha, \beta) : \alpha, \beta \in \F_{q^r} \right\}. 
$$
While (to our knowledge) at present there is not a closed form expression for $n_{q,r}$, values can be computed as in Table \ref{values}. Using the precise intersection numbers for odd $q$ produces higher dimensional codes.
\end{rmk}

We now provide examples to illustrate this fact. 

\begin{ex}
 Consider the curve $\mathcal X_{7,3}$ over $\F_{343}$. According to Table \ref{values}, $n_{q,r}=43$. 
 The curve-lifted code $C(D, \mathcal F_{\mathbb L_{343}, 43})$ includes codewords given by evaluating $x^ay^b$, $a+b \leq 41$, whereas the norm-trace-lifted code given by the bound $n_{7,3} \geq 16$ only considers those with $x^ay^b$, $a+b \leq 14$, according to Corollary \ref{no_spor_nt}.
\end{ex}
  
\begin{table}[h]
\begin{tabular}{|c|c|c|c|c|}
\hline
$p$ & $r$ &  $\#(L_{\alpha,\beta}\cap \mathcal{X}_{p,r})(\FF_{p^r})$& $B_{p,r}$ & $B_{p,r}'$ \\\hline 
 3& 2& 1, 4& 0& 0 \\ \hline
 3& 3& 13, 7, 10& 1& 4 \\ \hline
 3& 4& 22, 28, 31 & 8& 14 \\ \hline
 3& 5& 73, 76, 85, 91 & 38& 59 \\ \hline 
 3& 6&229, 244, 256& 152& 188 \\ \hline
 3& 7& 703, 715, 742, 757& 540& 634 \\ \hline
 5& 2& 1, 6& 0& 0 \\ \hline
 5& 3&  21, 26,31& 6& 10 \\ \hline
 5& 4& 111, 121, 126 141& 64& 64 \\ \hline
 5& 5& 561, 611, 621, 626, 641,  681& 445& 489 \\ \hline
 5& 6& 3056, 3106, 3126, 3131, 3206& 2624& 2724 \\ \hline
 5& 7& 15631, 15751, 15731, 15501, 15681, 15456& 14282& 14617 \\ \hline
 7& 2& 1, 8 & 0& 0 \\ \hline
 7& 3& 43, 50, 57& 16& 16 \\ \hline
 7& 4& 351, 358, 316, 379, 337& 216& 230 \\ \hline
 7& 5& 2451, 2325, 2465, 2381, 2437, 2395, 2353, 2402& 1955& 2029 \\ \hline
 7& 6& 16773, 16738, 17053, 16843, 16801, 16633, 16808& 15336& 15336 \\ \hline
 7& 7& \begin{tabular}{@{}c@{}} 117433, 117615, 118693, 117580,\\ 118063, 116173, 117895, 117853, 117685, 117643, 117475\end{tabular} & 112980& 113758 \\ \hline
 \end{tabular}
\caption{Collection of actual intersection numbers $n_{q,r}$ for all $\alpha, \beta$ and bounds}
\label{values}
\end{table}

\section{Conclusion} \label{conclusion_section}

In this paper, we defined curve-lifted codes which allow for local recovery by taking as repair groups the points of intersection of the curve with lines through the evaluation points. While inspired by Hermitian-lifted codes, they may exhibit different behaviors depending on the particular defining curve. We demonstrate when such codes have an explicit basis arising from typical monomial, unlike the Hermitian case where sporadic monomials are needed. In addition, we determined bounds on the number of affine points of intersection between a norm-trace curve and a line. We then used them to define norm-trace-lifted codes over fields of arbitrary characteristic. These new codes have high availability and positive rate, bounded away from zero as the code length goes to infinity. We note that codes over fields of small characteristic provide the best asymptotic rates. The opportunity to obtain greater rate by evaluating more functions may motivate one to consider more tailored bounds for the intersection numbers.

\end{document}